\documentclass[fleqn,twoside,11pt]{article}
\usepackage{amsfonts,amsmath,amssymb,amsthm}
\usepackage{bussproofs}
  \newcommand\Ax[1]{\AxiomC{$#1$}}
  \newcommand\UI[1]{\UnaryInfC{$#1$}}
  \newcommand\BI[1]{\BinaryInfC{$#1$}}
\usepackage{comment}
\usepackage[shortlabels]{enumitem}
  \setlist[enumerate,1]{leftmargin=22pt}
  \setlist[itemize,1]{leftmargin=25pt}
  \setlist[description,1]{leftmargin=15pt}
\usepackage[text={5.5in,10.2in},centering]{geometry}
\usepackage{graphicx}
\usepackage{hyperref}
\usepackage{mathabx}
\usepackage{quoting}
\usepackage{setspace}
\usepackage[dvipsnames]{xcolor}
  \newenvironment{Blue}{\noindent\color{Blue}}{}
  
  \newenvironment{Gray}{\noindent\color{Gray}}{}
  
  \newenvironment{Mah}{\noindent\color{Mahogany}}{}
  
  \newenvironment{Red}{\noindent\color{Red}}{}
  
\usepackage{url}

\newtheorem{theorem}{Theorem}[section]

\newtheorem{corollary}[theorem]{Corollary}
\newtheorem{lemma}[theorem]{Lemma}
\newtheorem{proposition}[theorem]{Proposition}

\theoremstyle{definition}

\newtheorem{definition}[theorem]{Definition}

\newtheorem{remark}[theorem]{Remark}

\renewcommand\a{\ensuremath{\forall}}
\newcommand\A{\ensuremath{\mathcal A}}
\newcommand\Bot{\ensuremath{\bot}}
\newcommand\bs{\ensuremath{\,\big/\,}} 

\newcommand\comma{\raisebox{2.5pt}{\,,}}
\newcommand\D{\ensuremath{\Delta}}
\renewcommand\d{\ensuremath{\delta}}
\newcommand\Dp{\ensuremath{\Delta'}}

\newcommand\entails{\vdash}
\newcommand\e{\ensuremath{\exists}}

\newcommand\false{\ensuremath{\mathtt{false}}}
\newcommand\fr[2]{\ensuremath{\displaystyle\frac{#1}{#2}}}
\newcommand\G{\ensuremath{\Gamma}}
\renewcommand\H{\ensuremath{\mathcal H}}

\newcommand\seq[1]{\ensuremath{\left\langle#1\right\rangle}}
\renewcommand\L{\ensuremath{\mathcal L}}
\newcommand\Land{\ensuremath{\land}}
\newcommand\lh[1]{\ensuremath{\text{lh}(#1)}}
\newcommand\Lor{\ensuremath{\lor}}
\newcommand\lo{\ensuremath{\L_1}}
\newcommand\loq{\ensuremath{\L_1^q}}
\newcommand\lqz{\ensuremath{\L_0^q}}
\newcommand\M[1]{\ensuremath{M\vDash\!{#1}}}
\newcommand\MEP[1]{\ensuremath{\text{MEP(#1)}}}
\renewcommand\models{\vDash}

\newcommand\omodels{\ensuremath{O\text{-models}}}
\newcommand\Pars[1]{\ensuremath{\text{Par}^*(#1)}}
\newcommand\period{\raisebox{2.5pt}{\,.}}
\newcommand\ph[1]{\smallskip\noindent\texttt{#1}} 
\renewcommand\phi{\ensuremath{\varphi}}
\newcommand\Q{\ensuremath{\mathcal Q}}
\newcommand\qd[1]{\ensuremath{\text{qd}(#1)}}
\newcommand\qef{\hfill$\triangleleft$} 
\newcommand\qefhere{\tag*{$\triangleleft$}} 
\newcommand\Quad{\mbox{}\quad}

\newcommand\s{\ensuremath{\bar S}}
\newcommand\said{\texttt{said}}
\newcommand\set[1]{\ensuremath{\left\{#1\right\}}}
\newcommand\Sub[1]{\ensuremath{\text{Sub}(#1)}}

\newcommand\To{\ensuremath{\to}}
\newcommand\Top{\ensuremath{\top}}
\newcommand\true{\ensuremath{\mathtt{true}}}

\title{Primal logic of information}
\author{Yuri Gurevich,\\
\normalsize Computer Science \& Engineering, 
            University of Michigan, USA
         \and
Andreas Blass, \\
\normalsize Mathematics, University of Michigan, USA}
\date{}

\begin{document}
\maketitle 

\begin{abstract}
Primal logic arose in access control; it has a remarkably efficient (linear time) decision procedure for its entailment problem.
But primal logic is a general logic of information.
In the realm of arbitrary items of information (infons), conjunction, disjunction, and implication may seem to correspond (set-theoretically) to union, intersection, and relative complementation.
But, while infons are closed under union, they are not closed under intersection or relative complementation.

It turns out that there is a systematic transformation of propositional intuitionistic calculi to the original (propositional) primal calculi; we call it Flatting.
We extend Flatting to quantifier rules, obtaining arguably the right quantified primal logic, QPL.
The QPL entailment problem is exponential-time complete, but it is polynomial-time complete in the case, of importance to applications (at least to access control), where the number of quantifiers is bounded.
\end{abstract}

\noindent
\texttt{Keywords}: logic of information, primal logic, linear time

\thispagestyle{empty}

\footnotetext{Partially supported by the US Army Research Office under W911NF-20-1-0297.}

\section{Introduction}\label{s:intro}

Primal logic works with items of information, \emph{infons} for brevity%
\footnote{The word infon was used earlier with a different meaning in situation semantics and related literature; see \cite{Devlin} for example.}.
An infon could be a number, a string, a piece of text, a picture, an audio, a video, etc.

One is tempted to view infons as containers of data and apply set-theoretic operations --- union, intersection, difference (a.k.a. relative complement) --- to infons.
Arguably infons are closed under union (just consider the two infons as one), but they are not closed under intersection or difference.
Indeed, try to extract the exact common information, i.e.\ all common data, in an audio and a still picture of the same event.
Even in the case of strings, the common information is problematic: There may be no string representing the common information of two given strings, according to algorithmic information theory \cite[Chapter~11]{Shen}.
Similarly, there is in general no way, given two infons $a$ and $b$ to extract the exact information in $b$ that isn't in $a$.

Primal logic is a general logic of infons.
It was discovered in the course of access control research and it was applied to access control \cite{G191,G198,G200,G203}.
But, from the very beginning, the discoverers viewed primal logic as a general logic of information: ``Here we investigate infon logic and a narrow but useful primal fragment of it'' \cite[Abstract]{G198}.
The fact that primal logic was discovered in access control is, we believe, an accident of history;
the role of access control could have been played by another field, e.g. privacy; in this connection see Section ``Personal infoset'' in \cite{G222}.

The original primal logic is propositional \cite{G198}.
It uses the traditional propositional nullary connective true (\Top) and binary conjunction (\Land), and implication (\To).
Article \cite{G215}, our main reference on propositional primal logic, also uses disjunction.
Intuitively, \Top\ is a noninformative infon.
The conjunction $a \land b$ of infons is their union which is rather natural.
For example, you know the union of $a,b$ if and only if you know both $a$ \emph{and} $b$.
The disjunction $a\lor b$ approximates the intersection of $a$ and $b$.
The implication $a\to b$ approximates the relative complement of $a$ in $b$, which is the smallest information that, when added to $a$, entails $b$.
It is required that every infon entails \Top, that $a$ and $b$ entail $a\land b$ and the other way round, that $a\lor b$ is entailed by either disjunct, and that $a$ and $a\to b$ entail $b$ which entails $a\to b$.
It is not required that $a\lor b$ is a most informative infon entailed by $a$ and entailed by $b$.
Nor is it required that $a\to b$ is a least informative infon which, together with $a$, entails $b$.

While primal logic was the most expressive access control logic at the time of its introduction, it is weak compared to traditional propositional logics.
The main technical result in \cite{G198} is that the entailment problem and even the multi-entailment problem for the original primal logic are solvable in linear time.
The \emph{multi-entailment problem} asks which of the given queries (alleged conclusions) follow from the given hypotheses; when only a single query is involved, we speak of the \emph{entailment problem}.

Requiring that primal logic be closed under genuine implication, i.e.\ restricting attention to scenarios where the infons are closed under set-theoretic difference, results in the entailment problem for primal logic becoming polynomial space complete \cite[Theorem~4.1]{G198}.
Requiring closure under genuine disjunction results in the entailment problem becoming co-NP complete \cite[Theorem~4.2]{G205}.

The original primal logic does not satisfy the replacement of equivalents principle.
For example, the formulas $A \land B$ and $B \land A$ are equivalent (each entails the other) but $(A\land B) \to C$ and $(B\land A) \to C$ are not.
Imposing the replacement of equivalents principle makes the entailment problem NP-hard \cite{Lev}.

But there are ways to enrich primal logic gently, so that the complexity of the entailment problem remains low.
Article \cite{G221} imposes the replacement of equivalents principle in the case where the equivalents are conjunctions of the same set of conjuncts.
This is achieved by identifying a conjunction with the set of its conjuncts, and reasoning accordingly.
In particular, the formulas $A \land B$ and $B \land A$ are identified and so are the formulas $\left(A \land B\right) \to C$ and $\left(B \land A\right) \to C$.
The multi-entailment problem for that extension is decidable in linear expected time and quadratic worst-case time.

Another gentle extension of the original primal logic adds the transitivity law for implication.
The law fails in the original primal logic: $A\to B$ and $B\to C$ do not necessarily entail $A\to C$.
The multi-entailment problem for transitive primal logic is decidable in quadratic time \cite[Theorem~5.1]{G211}

Extensions of (variants of) the original primal logic with the standard intuitionistic/classical quantifier rules were studied in \cite{Podgaits}.
Most attention is given to two extensions, one in the spirit of intuitionistic logic and the other in the spirit of classical logic.
In both cases, a cut elimination theorem is proven and an adequate semantics is given.
But the entailment problems for all those extensions are undecidable;
this follows from Theorem~\ref{t:undec} in Appendix~C below.

We find primal logic fascinating. It is rich enough for real-world applications, it is weak enough to admit remarkably efficient decision algorithms for the entailment problem, and it seems to be the right (or mostly right) logic about general information.

The collective experience of the logic community demonstrates that classical and intuitionistic logics are the right logics for certain known purposes.
In the case of primal logic, we don't have such experience,
and of course, the thesis that primal logic is the (mostly) right information logic cannot be proven mathematically.

What can be done theoretically at this stage?
Well, we found one argument that primal logic isn't ad hoc.
In \S3, we present a simple systematic transformation, called Flatting, of the propositional natural-deduction intuitionistic calculus into (a trivial extension of) the original primal calculus.
The trivial extension comprises an additional axiom $A\to A$ and an additional inference rule $\frac{A\lor A}A$.
The experimental provenance of the original primal calculus may explain why it misses these additions; they are obviously of little importance in applications.

Flatting can be applied also to quantifier rules and thus suggests the arguably right quantifier extension of propositional primal logic.
We call that extension QPL, quantified primal logic.

We quickly describe the rest of this paper.
In \S\ref{s:prelim} we recall some standard definitions in the form convenient for our purposes in this paper, and in \S\ref{s:orig} we revisit the original primal logic.
QPL is introduced in \S\ref{s:quant}.
In \S\ref{s:lin}, we prove that the multi-entailment problem for the propositional fragment of QPL is solvable in linear time.

\S\ref{s:local} is a key technical section. We prove that, if hypotheses \H\ entail a conclusion $C$ in QPL, then there is a QPL derivation of $C$ from \H\ with a strong subformula property.
In \S\ref{s:sem}, we develop a certain semantics for QPL and prove the soundness and completeness theorems.

Finally, in \S\ref{s:upper} and \S\ref{s:lower}, we prove that the entailment problem for QPL is exponential-time complete, but it is polynomial-time complete in the case, of importance to access control, where the number of quantifiers is bounded.

\section{Preliminaries}
\label{s:prelim}

We recall various definitions, mostly to establish terminology.

\subsection{Terms and formulas}\label{sb:formulas}

We use the language of first-order logic without equality or function symbols of positive arity.
Individual constants (in short, constants) are viewed as nullary function symbols.

\ph{Terms.}
Individual variables (in short, variables) and constants are terms. There are no other terms.

\ph{Atomic formulas.}
There are two kinds of atomic formulas. One kind comprises the two nullary propositional connectives \Top\ and \Bot, also known as \true\ and \false\ respectively.
The other atomic formulas have the form $R(t_1,\dots,t_j)$ where $R$ is a $j$-ary relation symbol and each $t_i$ is a term.
The arity $j$ may be zero, in which case $R$ may be called a \emph{propositional variable}.

\ph{Formulas.}
Formulas are built from atomic formulas by means of binary propositional connectives $\land, \lor, \to$ and quantifiers \a, \e.
If $A,B$ are formulas then so are $(A\land B)$, $(A\lor B)$, and $(A\to B)$;
we often omit the parentheses but in principle they are there to ensure the uniqueness of parsing.
If $A$ is a formula and $x$ a variable, then $\a x A$ and $\e x A$ are formulas.
Negation does not appear among the propositional connectives above, but a formula $A\to \bot$ may be abbreviated to $\neg A$.

An occurrence of a variable $x$ in a formula $A$ is bound if  it is in a quantification $qx$ ($\a x$ or $\e x$) or in the scope of a quantification; otherwise the occurrence is free.
A variable $x$ is free in a formula $A$ if there are free occurrences of $x$ in $A$.
A quantification $qx B$ is \emph{trivial} if $x$ isn't free in $B$.

\ph{Substitutions.}
To \emph{substitute} a term $t$ for a variable $x$ in a formula $A(x)$ means to replace every free occurrence of $x$ in $A(x)$ with $t$; the result is a formula $A(t)$.
This leads to a \emph{clash of variables} if $t$ is a variable and an occurrence of it becomes bound as a result of the substitution.
A term $t$ is \emph{substitutable} for $x$ in $A(x)$ if the substitution does not lead to a clash of variables.
When we use the notation $A(t)$, we tacitly assume substitutability.

\ph{Subformulas.}
The notion of \emph{subformula} is defined inductively.
\begin{enumerate}
\item $A$ is a subformula of $A$.
\item If $B\land C$, $B\lor C$, or $B\to C$ is a subformula of $A$, then so are $B$ and $C$.
\item If $\a xB(x)$ or $\e xB(x)$ is a subformula of $A$ then, so is $B(t)$ for any (substitutable) term $t$.
\end{enumerate}

\emph{Literal subformulas} are defined similarly, except that in the quantifier clause $t$ must be $x$.

Note that any subformula of a subformula of $A$ is itself a subformula of $A$, that is, the ``subformula'' relation is transitive.
The same applies to literal subformulas.

\ph{Parameters.}
A \emph{parameter} of a formula $A$ is a constant that occurs in $A$ or a variable that is free in $A$.
If $P$ is a set of constants and variables, then a \emph{$P$-formula} is a formula  all of whose  parameters are in $P$.

\subsection{Entailment}
\label{sb:entail}

When speaking about logics, we presume that, for any logic \L, the following two notions are well defined.
One is the notion of formula; for algorithmic reasons, it is required that formulas are strings in a finite alphabet.
The other notion is the binary relation $\H\vdash Q$ (in words, \H\ \emph{entails} $Q$) where \H\ is a set of formulas (the hypotheses) and $Q$ is a formula (the query).
It is required that entailment is a (Tarskian) finitary consequence relation, meaning that $\vdash$ is reflexive, monotone, transitive, and finitary:
\begin{itemize}
\item if $Q\in\H$ then \H\ entails $Q$,
\item \H\ entails $Q$ if some subset of \H\ entails $Q$,
\item if \H\ entails $Q$ and $\H'$ entails every formula in \H\, then $\H'$ entails $Q$, and
\item if \H\ entails $Q$ then some finite subset of \H\ entails $Q$.
    \end{itemize}
In this connection, see \cite[\S4]{CM} and historical references there.

\begin{definition}
The \emph{entailment problem} for a logic \L\ is the problem to decide, given a finite set \H\ of hypotheses and a query $Q$, whether \H\ entails $Q$ in \L. \qef
\end{definition}

\subsection{Algorithmics}

We use the standard computation model of the analysis of algorithms, and make usual assumptions about syntax, parsing, etc.
The details are spelled out in \S5 of \cite{G215}, but we expect that the reader may need those details only in connection with linear time.

For purposes of algorithmics, a set will be represented by any sequence listing its members.
For example, the set \set{A,B} would normally be represented by \seq{A,B} or \seq{B,A}, but can be also represented by \seq{A,B,B,A}, etc.

\subsection{Trees}

It is common in computer science that trees grow downward in the sense that the root is at the top, which is convenient to describe tree algorithms that start at the root and move down to the leaves.
In contrast, standard derivation trees in logic grow upward, which is natural from the information flow point of view: you deduce from axioms and hypotheses (leaves of the tree) down to the conclusion (the root).
While downward growing trees are used in our algorithmic references, in this paper by default our trees grow upward because we use them almost exclusively as derivation trees.

A \emph{tree} is a finite partially ordered set, whose elements are called \emph{nodes}, such that
\begin{itemize}
\item for every node $a$, the nodes $b$ with $b\le a$ are linearly ordered, and
\item there is a least node called the \emph{root} of the tree.
\end{itemize}

If $a\le b$ then $a$ is a \emph{descendant} of $b$ while $b$ is an \emph{ancestor} of $a$.
If $a<b$ then $a$ is \emph{lower} than $b$ and is a \emph{proper descendant} of $b$ while $b$ is \emph{higher} than $a$ and is a \emph{proper ancestor} of $a$.

If $a$ is a proper descendant of $b$ and there is no node in between, then $a$ is the \emph{child} of $b$, and $b$ is a \emph{parent} of $a$.
Nodes $a_1, a_2, \dots, a_n$ form a \emph{path} if $a_{i+1}$ is a child of $a_i$ for all $i<n$ or if $a_i$ is a child of $a_{i+1}$ for all $i<n$.
In the first case, the path $a_1 > a_2 > \cdots > a_n$ is \emph{descending}.
In the second case, the path $a_1 < a_2 < \cdots < a_n$ is \emph{ascending}.

Thus, children are below their parents, a parent has only one child, but a child may have more than one parent.

\subsection{Hilbert-style calculi}
\label{sb:hilbert}

We take a \emph{Hilbert-style calculus} to be given by rules of inference.
In turn, an inference rule is given by a finite sequence of premises and a single conclusion, where the premises and conclusion are formulas.
It is common to write rules in the form
\[ \fr{A_1, \dots, A_k} B \]
where $A_1, \dots, A_k$ are the premises and $B$ the conclusion.
We will use this common form but also a convenient alternative form
\[ \seq{A_1, \dots, A_k} \bs B \]
where the angle brackets are omitted if $k=1$.

Typically, rules are instances of rule schemas.
For example, the rule schema $\seq{A,B} \bs A\land B$ designates the set of all rules where the premises are two arbitrary formulas and the conclusion is their conjunction.
The example illustrates the most common rule schemas where the letters in the schema represent arbitrary formulas.
In some cases, however, there are restrictions.
For example, the rule schema $\a x A(x) \bs A(t)$ requires that $t$ be substitutable for $x$ in $A(x)$.
The conjunction example also illustrates that, while the letter premises of the schema are distinct, the corresponding formulas may not be.
For brevity, we use the term ``rule'' for both the schema and its instances; the meaning will be clear from the context.
Rules with no premises are \emph{axioms}, and rules with one or more premises are \emph{proper}.


\begin{definition}\label{d:dtree1}
A \emph{derivation} \D\ in a given Hilbert-style calculus is a finite tree whose leaves are partitioned into \emph{axiom nodes} and \emph{hypothesis nodes}; the other nodes are \emph{(proper) rule nodes}.
Every node $a$ is labeled with a formula $L(a)$, and every nonleaf node $a$ is also labeled with an inference rule $R(a)$.
It is required that
\begin{itemize}
\item the label $L(a)$ of any axiom node $a$ is an axiom, and
\item if a node $b$ has $k>0$ parents $a_1, \dots, a_k$, then
\[ R(b) = \seq{L(a_1), \dots, L(a_k)}\bs L(b). \qefhere \]
\end{itemize}
\end{definition}

In a derivation \D, the node labels are \emph{\D-formulas}.
The labels of the axiom nodes, the hypothesis nodes, and the root are the \emph{\D-axioms, \D-hypotheses}, and \emph{\D-conclusion} respectively.
We say that \D\ is a derivation of a formula $C$ from a set of formulas \H\ if $C$ is the conclusion of \D\ and \H\ includes all \D-hypotheses.
The derivation \D\ is \emph{minimal} if there is no derivation $\D'$ with the same conclusion but fewer nodes such that all $\D'$-hypotheses are \D-hypotheses.

We say that a logic \L\ \emph{admits} a Hilbert-style inference rule $\seq{A_1, \dots, A_k} \bs B$ if, in \L, the premises $A_1, \dots, A_k$ entail the conclusion $B$.
\L\ admits a rule schema if it admits every rule in the schema.

\section{Original primal logic}
\label{s:orig}

Originally, primal logic was propositional.
In this section, by default, logics are propositional.

Primal logic originated in access control \cite{G191} and  was explicitly introduced in article \cite{G198} under the name ``primal infon logic''.
\emph{Infons} are items of information discussed below in this section.

Primal logic used only three traditional propositional connectives: \Top, \Land, and \To.
But it was multi-agent like epistemic modal logic.
In the latter, each agent $p$ gives rise to the unary epistemic connective $p$ \texttt{knows} $A$.
In \cite{G198}, agents were called \emph{principals} and each principal $p$ gave rise to two \emph{quotation} connectives: $p\ \said\ A$ and $p\ \texttt{implied}\ A$.
Soon the \texttt{implied} connectives were dropped.
In particular, they are not used in article \cite{G215}, our main reference on propositional primal logic.
In the present paper, we take a wider view on primal logic and deemphasize quotation connectives, but the technical results can be extended more or less routinely to include quotation connectives.

\subsection{Hilbert-style calculus}
\label{sb:orig}

Article \cite{G198} gives two calculi for primal logic, a natural deduction calculus and a Hilbert-style calculus.
Subsequent articles on primal logic worked primarily with  Hilbert-style calculi.
For brevity, the quotation-free version of the Hilbert-style calculus in \cite{G198} will be called the \emph{original primal calculus}.

\begin{center}
\textbf
{Original primal calculus}\mbox{}\\
\end{center}
\begin{align*}
\top\text{I} \hspace{20pt}
&\top
\\[15pt]
\land\text{I}\hspace{20pt}
&\fr{A\quad B}{A\land B}\hspace{50pt}
&\land\text{E}\hspace{20pt}
&\fr{A\land B}A\comma\ \fr{A\land B}B
\\[20pt]
\to\text{I}\hspace{20pt}
&\fr B{A\to B}\hspace{50pt}
&\to\text{E}\hspace{20pt}
&\fr{A\quad A\to B}B
\end{align*}

\subsection{Algebra of infons}
\label{sb:algebra}

Infon algebra provides a certain intuition behind the primal calculus.

\begin{quoting}[leftmargin=30pt]\noindent
``One may study the algebra of infons. Order $a \leq b$ if the information in $a$ is a part of that in $b$.
At the bottom of that partial order are uninformative infons carrying no information whatsoever.
There is a natural union operation $a + b$ on infons. You know $a+b$ if you know $a$ and you know $b$.
Infon $a+b$ is the least upper bound of $a$ and $b$. But one has to be careful with the information order because of the omniscience problem well known in epistemic logic.
The partial order is intractable.'' \cite[\S1]{G198}
\end{quoting}

We will not use the intractable partial order or its equivalence relation which is also intractable.
But we need some equivalence relation on infons to ensure for example that the union operation is associative.


It will be convenient to use lattice-theoretic terminology and notation, so in particular the union operation will be called the join operation.

The \emph{join} of two infons $a,b$ is the result $a\sqcup b$ of combining the two infons and viewing the result as a single infon.
We require that there is an uninformative infon 0 and that, for all infons $a,b,c$, we have:
\begin{equation}\label{semi}
\begin{aligned}
(a\sqcup b) \sqcup c &= a \sqcup (b\sqcup c),\\
 a\sqcup b &= b\sqcup a,\\
 a\sqcup a &= a,\\
 a\sqcup 0 &= a.
\end{aligned}
\end{equation}
Thus the join operation is associative, commutative, and idempotent; 0 is a neutral element for $\sqcup$.
In other words, infons form a join semilattice with zero.

As usual in semilattices, we define a partial order
\begin{equation}\label{order}
a\le b\quad\text{if}\quad a\sqcup b = b.
\end{equation}
In the rest of this section, order \eqref{order} is the default information order.


\begin{lemma}\label{l:join}
If $c\ge a$ and $c\ge b$ then $c\ge a\sqcup b$.
\end{lemma}

\begin{proof}
Using \eqref{semi} and \eqref{order}, we have
$(a\sqcup b)\sqcup c = a\sqcup (b\sqcup c) = a\sqcup c = c$.
\end{proof}

Recall that, in a join semilattice, a \emph{pseudocomplement $a*b$ of $a$ relative to $b$} is the least element $x$, if it exists, such that $a\sqcup x \ge b$.
In other words, for all $y$, we have $a*b \le y$ if and only if $b \le a\sqcup y$.
(See \cite[\S I.6]{Graetzer} where Gr\"atzer studies the dual notion: pseudocomplementation in a meet semilattice.)

We expand the join semilattice with additional binary operation $*$ that has some but not necessarily all properties of pseudocomplementation. We require that, for all infons $a,b$, we have
\begin{equation}\label{pseudo}
 a*b \le b \le a \sqcup (a*b).
\end{equation}
By \eqref{order}, the requirements \eqref{pseudo} can be equivalently restated in the form of equalities:
\begin{equation}\label{pseudo2}
 (a*b)\sqcup b = b,\quad b\sqcup a\sqcup(a*b) = a\sqcup(a*b).
\end{equation}
We call the binary operation \emph{weak pseudocomplement}; ``weak'' indicates that we do not require that $a*b$ be the least element satisfying \eqref{pseudo}.

In universal algebra, an algebra is a structure whose vocabulary consists of function symbols.
Let $V$ be the variety (or equational class) of algebras satisfying the equations \eqref{semi} and \eqref{pseudo2}.
Modulo these equations, terms constructed from variables and 0 by means of joins and weak pseudocomplements naturally form a free algebra $F$ of variety $V$.
Variables are the natural free generators of this term algebra $F$.

Every term $t$ can be viewed as a formula of the original primal logic where $0,\sqcup,*$ represent \Top, \Land, and \To, respectively.
One could use a special notation, say $t'$, for the formula version of $t$.
Instead, we rely on the context.
For example, $t_1, t_2$ are terms in the context $t_1 \ge t_2$ and formulas in the context $t_1\entails t_2$.
By default, entailment will refer to entailment in the original primal calculus.

\begin{theorem}[Matching theorem]\label{t:apt}
$s \entails t$ in the original primal calculus if and only if $s\ge t$ as terms.
\end{theorem}

\begin{proof}
First we prove the only-if implication by induction on the number $n$ of nodes in a given derivation of $t$ from $s$.
The case $n=1$ is obvious. Suppose that $n>1$.
Several cases arise depending on the rule $R$ at the conclusion node labeled with $t$.

\noindent
Case \Land E. $R$ is $(t_1\land t_2)\bs t_i$ where $i\in\set{1,2}$ and $t$ is $t_i$.
Invoking the induction hypothesis, we have
\[ s\ge t_1\sqcup t_2  \ge t_i = t. \]

\noindent
Case \Land I. $R$ is $\seq{t_1, t_2}\bs t_1\land t_2$ where $t = t_1 \sqcup t_2$.
By the induction hypothesis, $s\ge t_1$ and $s\ge t_2$.
By Lemma~\ref{l:join}, $s\ge t_1\sqcup t_2 = t$.

\noindent
Case \To E. $R$ is $\seq{t_1, t_1\to t}\bs t$.
By the induction hypothesis, $s\ge t_1$ and $s\ge t_1*t$.
By Lemma~\ref{l:join} and \eqref{pseudo}, we have
$s \ge t_1 \sqcup (t_1*t) \ge t$.

\noindent
Case \To I. $R$ is $t_2\bs (t_1\to t_2)$ where $t = t_1\to t_2$.
By the induction hypothesis and \eqref{pseudo}, we have
$s \ge t_2 \ge t_1*t_2 = t$.

Next we prove the if implication.
The \emph{mutual entailment}
\[ s \dashv\vdash t\quad\text{if}\quad s\entails t\text{ and } t\entails s \]
is an equivalence relation on the terms.
Thanks to \Land E and \Land I, it satisfies the join semilattice requirements \eqref{semi}.
Mutual entailment also satisfies \eqref{pseudo2}, i.e.
\begin{equation}
(a\to b) \land b \dashv\vdash b\ \text{ and }\
   b\land (a\land (a\to b)) \dashv\vdash  a\land (a\to b)
\end{equation}
where $a,b$ range over our terms.
Indeed, in addition to conjunction rules, one application of \To I suffices to establish the first relationship and one application of \To E suffices to establish the second relationship.

In the term algebra, equality is the least equivalence relation satisfying the requirements \eqref{semi} and \eqref{pseudo2}.
It follows that $s=t$ implies $s\dashv\vdash t$.

Now suppose that $s\ge t$.
Then $s = s\sqcup t$, $s\dashv\vdash s\land t$, and $s\entails s\land t \entails t$.
\end{proof}

In \cite{G215}, the original primal logic was extended with disjunction and the standard disjunction introduction rules; the multi-entailment problem remained linear time decidable.
In \S\ref{s:lin}, the resulting logic is extended with \Bot, with the standard \Bot\ elimination rule, and with a weak disjunction elmination rule.
Again, the multi-entailment problem remains linear time decidable.
Given those extensions, the corresponding extensions of infon algebra as well as the appropriate generalizations of the matching theorem are rather obvious.

\section{QPL, Quantified Primal Logic}
\label{s:quant}

In the first part of this section, we recall the standard intuitionistic natural-deduction calculus.
In the second part, we present a formal transformation of this intuitionistic calculus to a Hilbert-style calculus which we use as the calculus for our quantified primal logic.

We use \G\ and \D\ to denote finite sequences of formulas.

\subsection{Intuitionistic natural-deduction calculus}\mbox{}

We start with the propositional fragment.

\begin{center}
\textbf
{Propositional rules}\mbox{}\\
\end{center}

\begin{align*}
\top\text{I} \hspace{20pt} &\top
\\
\land\text{I}\hspace{20pt} &\fr{A\quad B}{A\land B}\hspace{50pt}
&\land\text{E}\hspace{20pt}
&\fr{A\land B}A\comma\ \fr{A\land B}B
\\[6pt]
\to\text{I}\hspace{20pt} &\fr{\raisebox{7pt}{$(A)$}B}{A\to B}
&\to\text{E}\hspace{20pt}
&\fr{A\quad A\to B}B
\\
\lor\text{I}\hspace{20pt} &\fr{A}{A\lor B}\comma\ \fr{B}{A\lor B}
&\lor\text{E}\hspace{20pt}
&\fr{A\lor B\quad
 \raisebox{7pt}{$(A)$}C\quad\raisebox{7pt}{$(B)$}C} C
\\[5pt]
&&\bot\text{E}\hspace{20pt} &\fr\bot A
\end{align*}

These inference rules are self explanatory except possibly for the rules \To I and \Lor E.
For example, the rule \Land I says that from any formulas $A$ and $B$ one can deduce $A\land B$ in one step.
It follows that if a sequence \G\ of formulas entails $A$ and entails $B$ then it entails $A\land B$.

The intended meaning of the $\to$I rule is that \G\ entails $A\to B$ if $A$, \G\ entails $B$.
The intended meaning of the $\lor$E rule is that $A\lor B$, \G\ entails $C$ if $A$, \G\ entails $C$ and $B$, \G\ entails $C$.
These meanings are given succinctly by the $\to$I and $\lor$E metarules
\[
  \fr{A,\G\implies B}{\G\implies (A\to B)}\comma\qquad
    \fr{A,\G\implies C\quad B,\G\implies C}{A\lor B,\G\implies C}
\]
respectively.
For our purposes, the notation $\D\implies F$ indicates that \D\ entails $F$, that is there is an intuitionistic derivation of $F$ from \D.

The experts will recognize these metarules as rules of the intuitionistic sequent calculus.
``The calculi of sequents,'' writes Prawitz in \S2 of Appendix~A of his book \cite{Prawitz}, ``can be understood as meta-calculi for the deducibility relation in the corresponding systems of natural deduction.'' That is exactly how we use sequent rules in this paper.

The metarules for the other natural-deduction rules are obvious.
For example, the $\land$I and $\land$E metarules are
\[ A,B \implies A\land B,\quad
   A\land B\implies A,\quad A\land B\implies B. \]

Now we address \emph{structural metarules}.
Notice that each single metarule above describes a one-step deduction in the intuitionistic calculus.
There is also the necessity to combine deductions.
For example, consider $\G = \seq{A, A\to B, B\to C}$. By $\to$E, the first two hypotheses entail $B$.
Using $\to$E again, you derive $C$.
The structural \emph{cut} metarule describes that sort of combining deductions:
\[
 \fr{\G\implies A\quad A,\G\implies B}{\G\implies B}\period
\]

An even more basic property of natural deduction is that a deduction from a sequence \G\ of hypotheses is also a deduction from an extended sequence  $\G,\Delta$.
The structural metarule expressing this property is called \emph{weakening}:
\[
 \fr{\G\implies A}{\G,\Delta\implies A}\period
\]
In the presence of the weakening metarule, the cut metarule is equivalent to a formally stronger metarule
\[
 \fr{\G\implies A\quad A,\Delta\implies B}{\G,\Delta\implies B}\period
\]

The \emph{exchange} metarule
\[ \fr{\G,A,B,\D \implies C} {\G,B,A,\D \implies C}
\]
allows us to view sequences as multisets in the sense that the ordering of the formulas does not matter.
In the presence of the exchange metarule, the \emph{contraction} metarule
\[ \fr {A,A,\G \implies C} {A,\G \implies C}
\]
allows us to view sequences as sets in the sense that repetitions do not matter.

Finally, there is an obvious structural metarule
\begin{equation}\label{id}
 A,\G \implies A
\end{equation}
which we call the \emph{identity} metarule.
Structural metarules are used tacitly in the rest of this section.

\begin{center}
\textbf
{Quantifier rules}\mbox{}\\
\end{center}
\begin{align*}
\a\text{I}\hspace{20pt} &\fr{A(t)}{\a v A(v)}
&&\a\text{E}\hspace{20pt}
\fr{\a v A(v)}{A(t)}\\[9pt]
\e\text{I}\hspace{20pt} &\fr{A(t)}{\e v A(v)}
&&\e\text{E}\hspace{20pt}
\fr{\e v A(v)\ \raisebox{10pt}{$(A(t))$}B}{B}
\end{align*}
subject to restrictions, to be explained shortly, on \a I and \e E.

\medskip
The rules \a E and \e I have obvious metarules:
\[ \a v A(v) \implies A(t) \quad\text{and}\quad
   A(t) \implies \e v A(v) \]
respectively.
But recall the tacit requirement that $t$ be substitutable for $v$ in $A(v)$.

The forthcoming restrictions on the rules \a I and \e E will implicitly involve the hypotheses \G\ used to derive the premises.
For example, if \G\ contains an atomic formula $P(t)$, neither rule should be applicable.

\begin{definition}\label{d:exposed}
A term $t$ is \emph{exposed} in a sequence \G\ of formulas if $t$ is a constant that occurs in \G\ or $t$ is a variable that occurs free in \G. \qef
\end{definition}

The \a I and \e E metarules are respectively
\[ \fr{\G\implies A(t)}{\G\implies \a v A(v)} \quad\text{and}\quad
   \fr{A(t),\G \implies B}{\e v A(v),\G \implies B} \]
where $t$ is not exposed in $\G,B,\a v A(v)$ (or equivalently in $\G,B,\e v A(v)$).

For convenient reference, the intuitionistic natural deduction calculus is summarized in Appendix~\ref{app:a}.
For brevity, we call this calculus the intuitionistic calculus in the rest of this section.

\subsection{Enter QPL}

In this section, we put forward a formal transformation of the intuitionistic calculus above into a Hilbert-style calculus that we call QPL, \emph{quantified primal logic}.
%
Call a sequent $\D\implies F$ \emph{trivial} if $F\in\D$.
Let $R$ range over the rules of the intuitionistic calculus above.

\medskip\noindent
\textbf{Flatting.}
Restrict every rule $R$ so that the premises, if any, of $R$'s metarule are trivial. \qef

Flatting is the formal transformation mentioned above, converting the intuitionistic calculus to QPL.

By virtue of the identity metarule \eqref{id}, Flatting makes the premises of the metarule redundant. As a result the new metarule has the Hilbert-style form $\G\implies F$, with no premises.

\begin{remark}
In the previous section, 8 out of 12 metarules are premise-less to begin with.
We could present all 12 metarules in the premises-conclusion format used in cases \To I, \Lor E, \a I, and \e E.
For example, the metarule for \Land I would be
\[ \fr{\G \implies A,\ \G\implies B} {\G \implies A\land B}. \]
In the presence of structural metarules, every premise-less metarule is equivalent to its premises-conclusion form, and that is why we used the simpler forms.
But, if we used the premises-conclusion form in all 12 cases, then Flatting would apply uniformly to all 12 of them.
In 8 cases, Flatting would recover the premise-less forms and thus would have no serious effect. \qef
\end{remark}

One propositional intuitionistic rule affected by Flatting is $\to$I.
The rule and the corresponding metarule are
\[\fr{\raisebox{7pt}{$(A)$}B}{A\to B}\qquad\text{and}\qquad
  \fr{A,\G\implies B}{\G\implies (A\to B)} \]
Flatting requires that $B = A$ or $B\in\G$.
If $B=A$, then the metarule becomes $\G\implies (A\to A)$ and the rule becomes an axiom $A\to A$. If $B\in\G$, then the metarule becomes $B,\D\implies (A\to B)$, and the rule becomes $B \bs A\to B$.

The only other propositional intuitionistic rule affected by Flatting is $\lor$E. The rule and the corresponding metarule are
\[\fr{A\lor B\quad
 \raisebox{7pt}{$(A)$}C\quad\raisebox{7pt}{$(B)$}C} C \qquad\text{and}\qquad
  \fr{A,\G\implies C\qquad B,\G\implies C}{A\lor B,\G\implies C} \]
Flatting requires that $C\in\G$ or else $C = A = B$. If $C\in\G$, then the new metarule $A\lor B, C, \D \implies C$ is a special case of the identity metarule \eqref{id}.
If $C = A = B$, then the new metarule is $C\lor C,\G \implies C$ and the new rule is $C\lor C\bs C$.

The intuitionistic quantifier rules affected by Flatting are \a I and \e E.
The $\a$I rule and metarule are
\[ \fr{A(t)}{\a v A(v)}\quad \text{and}\quad
 \fr{\G\implies A(t)}{\G\implies\a v A(v)}\]
where $t$ is not exposed in $\G, \a v A(v)$.
Flatting requires that $A(t)$ is in $\G$.
So $t$ is not exposed in $A = A(t)$, and therefore $v$ is not free in $A(v) = A$.
The new metarule is $A,\D \implies \a v A$ and the new rule is $A\bs \a v A$; in both cases, it is required that $v$ isn't free in $A$.
Thus, the new rule allows universal quantification only when the quantifier is vacuous.

The $\e$E rule and metarule are respectively
\[\fr{\e v A(v)\quad \raisebox{10pt}{$(A(t))$}B} {B}
 \quad\text{and}\quad
 \fr{A(t),\G\implies B}{\e v A(v),\G\implies B}\]
where $t$ is not exposed in  $\G, \e v A(v), B$.
Flatting requires that $B = A(t)$ or $B\in\G$.
If $B\in\G$, the metarule is a special case of the identity metarule \eqref{id}.
Suppose that $B = A(t)$, so that $t$ doesn't occur free in $A = A(t)$.
The new $\e E$ metarule is $\e v A,\G\implies A$ where $v$ isn't free in $A$ and the new $\e E$ rule is $\e v A\bs A$ where $v$ isn't free in $A$.
Again, we can eliminate existential quantifiers only when they are vacuous.

For convenient reference, QPL is summarized in Appendix~\ref{app:b}.

\subsection{Discussion}
\label{sb:disc}\mbox{}

\noindent\texttt{The choice of intuitionistic calculus.}There are other natural deduction calculi for intuitionistic logic, and applying Flatting to them may give inequivalent calculi.
For example, the calculus in Appendix~\ref{app:a} can be expanded with a redundant rule $\seq{(A\to B), (B\to C)}\bs (A\to C)$ which survives under Flatting but, by Corollary~\ref{c:trans} below, fails in QPL.
So what is special about the calculus in Appendix~\ref{app:a}?
It is the most standard natural deduction calculus for intuitionistic logic.
More importantly, it is the most straightforward and systematic implementation of Gentzen's idea that every logical operation is determined by an introduction rule and an elimination rule.

\smallskip\noindent\texttt{The name of the transformation.}
We considered several names, in particular ``Skip History'' and ``No Hypotheticals'' which are more descriptive than
``Flatting'' but ``Flatting'' is shorter.
Brevity is also the reason for using the archaic ``Flatting'' rather than ``Flattening.''

\section{Propositional fragment of QPL}
\label{s:lin}

In this section, all logics are propositional.
Let PfQPL be the propositional fragment of QPL.
We start with generalizing the entailment problem for a logic \L.

\begin{definition}
The \emph{multi-entailment problem} for a logic \L, in short \MEP{\L}, is the problem to decide, given a finite sequence of hypotheses and a finite sequence of queries, which of the queries are entailed by the hypotheses. \qef
\end{definition}

The generalization is trivial for complexity levels like PTime, but it is meaningful for finer complexity levels, in particular for linear time.
The main purpose of this section to establish that MEP(PfQPL) is linear time decidable.

The multi-entailment problem for the original primal logic is proven to be linear time decidable in \cite{G198}.
A slightly different version \lqz\ of primal logic is studied in article \cite{G215}.
Syntactically, \lqz\ formulas are built from atomic formulas and propositional constants \Top, \Bot\ by means of (a)~traditional binary propositional connectives \Land, \Lor, \To\ and (b)~unary \emph{quotational} connectives $p\ \said$ where $p$ is an individual constant.
There is a potentially infinite list of individual constants.
Intuitively they represent agents.
Logic \lqz\ is given by the following Hilbert-style calculus where $\pi$ may be any sequence
\[ p_n\ \said\;\ p_{n-1}\ \said\;\ \dots\ p_2\ \said\;\ p_1\ \said \]
of $n\ge0$ quotational connectives.

\begin{center}
\textbf
{\lqz\ calculus}\mbox{}\\
\end{center}
\begin{align*}
\top\text{I}\hspace{20pt}
&\pi\top
\\[1em]
\land\text{I}\hspace{20pt}
&\frac{\pi A\quad \pi B}{\pi(A\land B)}
\hspace{80pt}
&\land\text{E}\hspace{20pt}
&\frac{\pi (A\land B)}{\pi A}\comma \ \frac{\pi (A\land B)}{\pi B}
\\[1em]
\to\!\text{I}\hspace{20pt}
&\frac{\pi B}{\pi (A\to B)}
\hspace{80pt}
&\to\!\text{E}\hspace{20pt}
&\frac{\pi A\quad \pi (A\to B)} {\pi B}
\\[1em]
\lor\text{I}\hspace{20pt}
&\frac{\pi A}{\pi (A\lor B)}\comma \ \frac{\pi B}{\pi (A\lor B)}
\end{align*}

\smallskip
Let \loq\ be the extension of \lqz\ with this \Lor E rule:\quad $\fr{\pi (A\lor A)} {\pi A}$.

\begin{proposition}\label{p:0}
\MEP{\loq} is linear time decidable.
\end{proposition}

\begin{proof}
\MEP{\lqz} is proven linear time decidable in \cite[Theorem~5.1]{G215}.
The proof of that theorem exhibited the decision algorithm in detail.
That algorithm is mostly an elaboration of the linear time decision algorithm in \cite{G198}.

It is straightforward, for those familiar with that decision algorithm, to extend it to a linear time decision algorithm for \loq.
We need to augment \S5.6 of \cite{G215} with one additional case treating the new \Lor E rule.
The \Lor E case is similar to (and simpler than) the \Land E case.
\end{proof}

Let \lo\ be the quotation-free fragment of \loq.
The \lo\ calculus is obtained from the \loq\ calculus by erasing the quotation prefixes $\pi$.

\begin{corollary}\label{c:l1}
\MEP{\lo} is linear time decidable.
\end{corollary}

\begin{proof}
It suffices to check that quotation-free hypotheses \H\ entail a quotation-free query $Q$ in \lo\ if and only if \H\ entails $Q$ in \loq.
The only-if implication is trivial.
To prove the if implication, notice that erasing all quotation prefixes in an \loq\ derivation of $Q$ from \H\ results in an \lo\ derivation of $Q$ from \H.
\end{proof}

Let $\L_2$ be the extension of \lo\ with the standard \Bot E rule\quad $\fr\bot A$.

\begin{corollary}\label{c:l2}
\MEP{$\L_2$} is linear time decidable.
\end{corollary}

\begin{proof}
Given hypotheses \H\ and queries \Q, use the decision algorithm for \MEP{\lo} to check whether \H\ entails \Bot\ in \lo.
If yes, then all \Q\ queries are entailed by \H\ in $\L_2$.
Otherwise, use the decision algorithm for \MEP{\lo} to check which \Q\ queries are entailed by \H\ in \lo, and output the result.
To see that this is correct, consider an \lo\ derivation \D\ from \H.
\Bot\ cannot appear as the label of any node in \D, and so the \Bot E rule is not applied in \D.
Thus, \H\ entails the same consequences in \lo\ and in $\L_2$.
\end{proof}

Notice that the propositional fragment PfQPL of QPL is the extension of $\L_2$ with the axiom schema\quad $A\to A$.

Given a derivation \D\ in PfQPL and an axiom node with label $L$, we say that the axiom $L$ is \emph{local} to \D\ if $L=\Top$ or $L$ also occurs in a hypothesis or the conclusion.

\begin{lemma}\label{l:alocal}
If \D\ is a PfQPL derivation of $Q$ from \H\ that is minimal (according to \S\ref{sb:hilbert}), then all \D-axioms are local.
\end{lemma}

\begin{proof}
The proof is by induction on the number of nonlocal \D-axioms.
If all \D-axioms are local, we are done.
Otherwise pick a nonlocal \D-axiom $C\to C$ and replace all occurrences of $C\to C$ in \D\ with \Top.
As a result, every formula $X$ is transformed into a formula $X'$.
Let \Dp\ be the resulting labeled tree.
We prove that \Dp\ is a derivation of $Q$ from \H\ with the same number of nodes as \D\ but fewer nonlocal axioms.

Since $C\to C$ is nonlocal, the hypotheses and conclusion remain unchanged.
For the same reason, every local axiom $L$ remains unchanged.
Every nonlocal axiom $A\to A$ in \D\ is transformed into $A'\to A'$ unless $A=C$ in which case $A\to A$ is transformed into \Top.

It remains to check that every inference rule instance $R$ in \D\ is transformed into a valid inference $R'$ in \Dp.
If $R$ is
\begin{align*}
& \seq{A, B}\bs A\land B,\quad A\land B\bs A,\quad B\land A\bs A,\quad \\
& A\bs A\lor B,\quad B\bs A\lor B,\quad A\lor A \bs A,\quad
\text{or}\quad \bot\bs A
\end{align*}
then $R'$ is a valid inference respectively:
\begin{align*}
& \seq{A', B'}\bs A'\land B',\quad A'\land B'\bs A',\quad B'\land A'\bs A',\quad \\
& A'\bs A'\lor B',\quad B'\bs A'\lor B',\quad A'\lor A' \bs A',\quad
\text{or}\quad \bot\bs A'.
\end{align*}
Suppose that $R$ is
\[ B\bs A\to B\quad\text{or}\quad \seq{A, A\to B}\bs B. \]
If $A\ne C$ or $B\ne C$, then $R'$ is a valid inference
\[ B'\bs A'\to B'\quad\text{or}\quad \seq{A', A'\to B'}\bs B' \]
respectively.
It remains to consider the situations where $R$ is
\[ C\bs C\to C\quad\text{or}\quad \seq{C, C\to C}\bs C \]
in \D\, so that $R'$ is $C\bs\Top$ or $\seq{C,\Top}\bs C$.
But these situations are impossible because of the minimality of \D.
If $R$ is $C\bs C\to C$, remove the premise and the part of the deduction leading to it, and if $R$ is $\seq{C, C\to C}\bs C$, remove the premise $C\to C$ as well as the part of deduction leading to it and identify the premise $C$ with the conclusion $C$.
In either case, the result is a deduction of $Q$ from \H\ with fewer nodes.
\end{proof}

\begin{theorem}\label{t:star}
\MEP{PfQPL} is linear time decidable.

\end{theorem}

\begin{proof}
This theorem is a consequence of Corollary~\ref{c:l2} and the previous lemma.
In view of Corollary~\ref{c:l2}, it suffices to exhibit an algorithm that, given arbitrary sequences $\H,\Q$ of PfQPL hypotheses and queries, constructs in linear time a sequence $\H'$ of axioms of the form $A\to A$ such that, for every query $Q\in \Q$, $\H$ entails $Q$ in PfQPL if and only if $\H', \H$ entails $Q$ in $\L_2$.

By the lemma, the desired $\H'$ comprises the axioms $A\to A$ that occur in \H,\Q.
Given a little background in linear-time algorithms, it is easy to see that such an algorithm can be designed to work in linear time.
The necessary background is found, e.g., in an initial segment of \S5 of \cite{G215}.
\end{proof}

\section{Locality}
\label{s:local}

In Sections~\ref{s:local}--\ref{s:lower} we work with QPL.
Terms, formulas, inference rules, and derivations are by default those of QPL.

\subsection{Parameter locality}

Recall that a parameter of a formula $X$ is a constant that occurs in $X$ or a variable that is free in $X$.

\begin{definition}\mbox{}\label{d:par}
Fix an individual constant.
For a set $S$ of formulas,
\Pars{S} is the set $P_0$ of parameters in the formulas of $S$, unless $P_0$ has no constants, in which case \Pars{S} is the extension of $P_0$ with the fixed constant.  \qef
\end{definition}

Thus, \Pars{S} always contains a constant.
The star in the notation is intended as a reminder of this fact.

\begin{definition}\mbox{}\label{d:par}
Let \D\ be a derivation and $S$ the set comprising the hypotheses and the conclusion of \D.
A \emph{parameter} in \D\ is a parameter in a node label in \D.
A parameter $p$ in \D\ is \emph{local} (in \D) if $p\in\Pars{S}$.  \qef
\end{definition}

\begin{proposition}\label{p:plocal}
For every derivation \D, there is a derivation $\D'$ with the same hypotheses, the same conclusion, and the same number of nodes, such that all parameters in $\D'$ are local.
\end{proposition}

\begin{proof}
Replace every nonlocal parameter in \D\ with the same local constant $p$.
We shall use primes to indicate that replacement, writing, for example, $\D', X', R'(a)$.
It suffices to show that $\D'$, the resulting labeled tree, is a derivation.

If $X$ is a hypothesis, the conclusion, or the axiom \Top, then $X' = X$. If $X$ is an axiom $A\to A$, then $X'$ is the axiom $A'\to A'$.
It remains to prove the claim that, for every nonleaf node $a$, $R'(a)$ is a valid inference.

The claim is obvious if rule $R(a)$ is propositional, or \a I, or \e E.
For example,\\
if $R(a) = \seq{A, B}\bs A\land B$, then $R'(a) = \seq{A', B'}\bs A'\land B'$, and\\
if $R(a) = \e x A\bs A$, then $R'(a) = \e x A'\bs A'$ because $x$ is not free in $A$.

Suppose that $R(a)$ is an instance $\a x A(x)\bs A(t)$ of \a E. If parameter $t$ is local, then $R'(a)$ is the instance $\a x A'(x)\bs A'(t)$ of \a E.
Otherwise $R'(a)$ is the instance $\a x A'(x)\bs A'(p)$ of \a E.

The case where $R(a)$ is an instance $A(t)\bs\e x A(x)$ of \e I is similar.
\end{proof}

\subsection{The closure of a set of formulas}\mbox{}
\label{s:closure}

The subformulas and literal subformulas of a given formula $A$ are defined in \S\ref{sb:formulas}.
Let $P$ be a set of parameters.
A \emph{$P$-formula} is a formula with all parameters in $P$.
A \emph{$P$-subformula} of a formula $A$ is a subformula of $A$ that is a $P$-formula.
Note that the ``$P$-subformula'' relation, like the ``subformula'' relation, is transitive.

\begin{lemma}\label{l:singleformula1}
Each $P$-subformula $X$ of a $P$-formula $A$ is obtained from a literal subformula $B$ of $A$ by replacing some occurrences (possibly none)  of free variables with parameters in $P$.
More exactly, if $B$ is in the scope of a quantification $qx$ (\,$\a x$ or $\e x$) then every occurrence of $x$ in $B$ is replaced with the same parameter in $P$.
\end{lemma}

\begin{proof}
Induction on the number $n$ of binary connectives and quantifications in $A$.
The claim is trivial if $X = A$ and, in particular, if $n=0$ so that $A$ is atomic.
From now on, we assume that $X\ne A$.

Suppose that $A = B*C$, where $*$ is a binary propositional connective.
Since $X\ne A$, it is a subformula of $B$ or of $C$; use the induction hypothesis.
Suppose that $A = qx B(x)$.
Since $X\ne A$, it is a subformula of $B(t)$ for some $t\in P$; use the induction hypothesis.
($X$ can be a subformula of $B(x)$ but then, since $X$ is a $P$-formula, $x$ is in $P$ or else $x$ does not occur in $X$ in which case $X$ is a subformula of any $B(t)$ with $t\in P$.)
\end{proof}

By induction on formula $A$ we define the \emph{quantifier depth} \qd{A} of $A$.
\begin{itemize}
\item If $A$ is atomic then $\qd{A} = 0$.
\item If  $A$ is $B\land C$, $B\lor C$, or $B\to C$ then $\qd{A} = \max\set{\qd{B},\qd{C}}$.
\item If $A$ is $\a x B(x)$ or $\e x B(x)$ then $\qd{A} = 1 + \qd{B(x)}$.
\end{itemize}

Note that any literal subformula $B$ of $A$ is in the scopes of at most $\qd{A}$ quantifiers, because those scopes are nested in $A$.

For algorithmic purposes, a formula $A$ is a string of symbols.
The length \lh{A} of $A$ is the total number of (occurrences of) symbols in $A$.
Think of $A$ as a sequence of \lh{A} positions, and at each position there is a symbol.

\begin{lemma}\label{l:singleformula2}
The number of $P$-subformulas of a $P$-formula $A$ is at most $\lh{A}\cdot|P|^{\qd{A}}$.
\end{lemma}

\begin{proof}
By the previous lemma, the number of $P$-subformulas of $A$ obtained from a literal subformula $B$ of $A$ is at most $|P|^{\qd{A}}$.
Each literal subformula is a substring of $A$ starting at a certain position in string $A$.
Different literal subformulas start at different positions in $A$.
Thus $A$ has at most $|A|$ literal subformulas, and the claim follows.
\end{proof}

Let $S$ be a finite set of formulas, and let $P = \Pars{S}$.
$A$ is a \emph{subformula of} $S$ if it is a subformula of a formula in $S$.

\begin{definition}\label{d:close}
The \emph{closure} \s\ of $S$ is the set of $P$-subformulas of $S$.
\qef
\end{definition}

Let $n = \lh{S}$ be the length of $S$ meaning the sum of the lengths of its formulas.
Let $d = \qd{S}$ be the \emph{quantifier depth} of $S$ meaning $\max\set{\qd{A}: A\in S}$.

The length \lh{\s} may be exponential in $n$.
For example, if $S$ comprises a single formula
\[\a x_1\a x_2\dots\a x_r R(p_1,x_1,p_2,x_2,\dots,p_r,x_r) \]
then $P$ comprises the $r$ parameters $p_i$, $d = r$, and $n = O(r)$, so that $\lh{\s} \ge n r^r$ is exponential in $n$.

\begin{proposition}\label{p:closuresize}
\begin{enumerate}\mbox{}
\item The cardinality $|\s| \le n\cdot |P|^d$ and therefore $\lh{\s} \le n^2\cdot|P|^d$.
\item Asymptotically $\lh{\s} \le n^{\frac n2}$.
\item If we restrict attention to sets $S$ of bounded quantifier depth, then $\lh{\s}$ is bounded by a polynomial in $n$.
\end{enumerate}
\end{proposition}

\begin{proof}
(1) follows from Lemma~\ref{l:singleformula2}.
(2) follows from (1) because $|P|<n$ and $d<n/2$.
(3) follows from (1).
\end{proof}

\begin{remark}
In a nontrivial quantification $qx B(x)$, the variable $x$ appears at least once in scope $B(x)$ of $qx$ and is followed there by a comma or right parenthesis.
Restricting attention to nontrivial quantification, we can improve the exponent $n/2$ in (2) to $n/4$. \qef
\end{remark}

\subsection{Minor and major parents}

Most QPL rules have just one premise.
The two exceptions are \Land I and \To E which have two premises.
While the two premises of \Land I play similar roles, the premises of \To E play different roles.

\begin{definition}\label{d:major}\mbox{}
\begin{itemize}
\item The premises $A$ and $A\to B$ of the rule $\seq{A, A\to B}\bs B$ are the \emph{minor} and \emph{major} premises respectively.
\item A parent $b$ of a node $a$ in a derivation is the \emph{minor} parent of $a$ if the rule $R(a)$ at node $a$ is \To E and the label of $b$ is the minor premise of the rule.
    In all other cases $b$ is a \emph{major} parent of $a$.
\end{itemize}
\end{definition}

For future reference, we formulate an obvious corollary.

\begin{corollary}\label{c:major}
Every nonleaf node $a$ has a major parent, and if the rule at $a$ is an elimination rule then $a$ has a unique major parent.
\end{corollary}

\begin{lemma}\label{l:ei}
If $b$ is a major parent of a node $a$ in a minimal derivation and  $R(b)$ is an introduction rule then $R(a)$ is an introduction rule as well.
\end{lemma}

The condition that the parent $b$ is major can't be waived as the following counterexample shows.

\begin{center}
\begin{tabular}{ccc}
\Ax{A}
\UI{A\lor B}
\Ax{(A\lor B)\to C}
\BI{C}
\DisplayProof
\end{tabular}
\end{center}

\begin{proof}
Suppose toward a contradiction that $R(a)$ is an elimination rule.
Since $b$ is the major parent of $a$, the logical operation introduced by $R(b)$ is the same one that is eliminated by $R(a)$.
We claim that there is a parent $c$ of $b$ such that $L(c) = L(a)$, so that we can replace the subderivation above (and including) $a$ with the subderivation above $c$, obtaining a smaller derivation with the same conclusion and at most the same hypotheses as \D, which contradicts the minimality of \D.

The proof of the claim is routine and splits into five cases depending on the logical operation introduced by $R(b)$.
\begin{itemize}[leftmargin=10pt]
\item Case $\seq{A,B}\bs A\land B$, so that $R(a)$ is $A\land B\bs A$ or $A\land B\bs B$.
In either case we have $L(c) = L(a)$ for at least one parent $c$ of $b$.
\item Cases $A\bs A\lor B$ and $B\bs A\lor B$. We have only one \Lor E rule, namely $A\lor A\bs A$, so $A = B$, and the desired $c$ is the parent of $b$.
\item Case $B\bs A\to B$, so that $R(a)$ is $\seq{A, A\to B}\bs B$, and the desired $c$ is the parent of $b$.
\item Case $A\bs \a x A$ where $x$ isn't free in $A$, so that $R(a)$ is $\a x A\bs A$, and the desired $c$ is the parent of $b$.
\item Case $A(t)\bs \e x A(x)$, so that $R(a) = \e x A(x)\bs A(x)$ and $x$ isn't free in $A(x)$, and therefore $A(t) = A(x)$; the desired $c$ is the parent of $b$. \qedhere
\end{itemize}
\end{proof}

\begin{corollary}\label{c:unique}
For every node $a$ in a derivation, there is an ascending path $a = a_0 < a_1 < \dots < a_n$ where every $a_{i+1}$ is a major parent of $a_i$ and $a_n$ is a leaf; here $n$ may be zero.
If $R(a)$ is an elimination rule, then the path is unique and every $R(a_i)$ is an elimination rule.
\end{corollary}

\subsection{Formula locality}\mbox{}
\label{sb:alocal}

\begin{definition}\label{d:alocal}
Let \D\ be a derivation and $S$ the set comprising the hypotheses and the conclusion of \D.
A \D-formula is \emph{local} to \D\ if it belongs to the closure \s\ of $S$.
\end{definition}

\begin{lemma}\label{l:alocal2}
For every derivation there is a derivation with the same hypotheses and conclusion where all axioms are local.
\end{lemma}

\begin{proof}
The proof is similar to that of Lemma~\ref{l:alocal}.
Let \D\ be a derivation of $Q$ from \H\ that is minimal, $S = \H\cup\set{Q}$, and $P = \Pars{S}$.
By Proposition~\ref{p:plocal}, we may assume that every parameter in \D\ is local and thus belongs to $P$.

The proof is by induction on the number of nonlocal axioms in \D.
If all axioms in \D\ are local, we are done.
Otherwise pick a nonlocal axiom $\alpha = C\to C$.
Construct the least set $I$ of formulas such that (a)~$\alpha\in I$ and (b)~if $x,p\in P$, $x$ is a variable, and $\beta(x)\in I$, then $\beta(p)\in I$.
Members of $I$ will be called \emph{instances} of $\alpha$ or \emph{$\alpha$-instances}.

Replace all occurrences of every $\alpha$-instance in \D-formulas $X$ with \Top; as a result, $X$ is transformed into a formula $X'$.
Let \Dp\ be the resulting labeled tree.
We prove that \Dp\ is a derivation with the same hypotheses, the same conclusion, and the same number of nodes as \D, but with fewer nonlocal axioms.

By Definition~\ref{d:alocal}, all local \D-formulas are in \s. By the transitivity of the ``$P$-subformula'' relation, all local \D-formulas remain unchanged.
In particular, the hypotheses, the conclusion,  and the local axioms remain unchanged.
Every nonlocal axiom $A\to A$ is transformed into an axiom $A'\to A'$ unless $A\to A$ is an instance of $\alpha$, in which case $A\to A$ is transformed into \Top.

It remains to verify that every node rule $R = R(a)$ in \D\ is transformed into a valid inference rule $R'$ in \Dp.
The verification splits into several cases depending on what kind of rule $R$ is.
The propositional cases are treated exactly as in the proof of Lemma~\ref{l:alocal}.
Here we consider the quantifier cases.

If $R$ is $A\bs \a x A$ or $\e x A\bs A$ then $x$ isn't free in $A$ and $R'$ is a valid inference rule $A'\bs \a x A'$ or $\e x A'\bs A'$ respectively.
Suppose that $R$ is $\a x A(x)\bs A(t)$ or $A(t)\bs \e x A(x)$.
The case where $x$ is not free in $A(x)$ is handled exactly like the cases $A\bs \a x A$ and $\e x A\bs A$ above.
So, assume that $x$ is free in $A(x)$ and therefore $t$ occurs in $A(t)$.
Since all parameters in \D\ are local, $t\in P$.
If an $\alpha$-instance $\beta(x)$ occurs in $A(x)$, it becomes an $\alpha$-instance $\beta(t)$ in $A(t)$.
Accordingly, if $B(x) = (A(x))'$ then $(A(t))' = B(t)$.
Thus $R'$ is a valid inference rule $\a x B(x)\bs B(t)$ or $B(t)\bs \e x B(x)$ respectively.
That completes the verification and the proof of the lemma.
\end{proof}

\begin{lemma}
In a minimal derivation, the premises of any elimination rule are local.
\end{lemma}

\begin{proof}
Since \To E is the only elimination rule with a minor premise and
the minor premise of \To E is a subformula of the major premise, it suffices to prove that the major premise of any elimination rule is local.

By Corollary~\ref{c:unique}, for every elimination-rule node $a$, there is a unique path $a = a_0, a_1, \dots, a_d$ of elimination-rule nodes from $a$ to a leaf where every parent is major.
The proof of the lemma is by induction on the length $d$ of that path.
Let $b$ be the major parent of $a$ and $P = L(b)$.
We need to prove that $P$ is local.

\smallskip\noindent
Induction base: $d=1$.
In this case, $b$ is a leaf and therefore $P$ is a hypothesis or axiom.
By the definition of locality and Lemma~\ref{l:alocal}, $P$ is local.

\smallskip\noindent
Induction step: $d>1$.
By the induction hypothesis, the major premise of $R(b)$ is local.
Since $R(b)$ is an elimination rule, it follows that the conclusion $P = L(b)$ of $R(b)$ is local.
\end{proof}

\begin{theorem}\label{t:local}[Locality theorem]
In a minimal derivation, all formulas local.
Thus, if \D\ is a minimal derivation and $S$ the set comprising the hypotheses and the conclusion of \D, then all \D-formulas belong to the closure \s.
\end{theorem}

\begin{proof}
By the definition of locality, the conclusion is local.
It suffices to prove that, if $L(a)$ is local and $b$ is a parent of $a$, then $L(b)$ is local.
If $R(a)$ is an introduction rule, the claim is obvious.
If $R(a)$ is an elimination rule, use the previous lemma.
\end{proof}

\section{QPL semantics}
\label{s:sem}

\begin{definition}\label{d:standard}
Let $S$ be a set of formulas and $P =\Pars{S}$.
\begin{itemize}
\item A \emph{standard} structure for $S$ is a (first-order) structure $M$ subject to the following requirements.
 \begin{itemize}
 \item The universe (or base set) of $M$ is $P$,
 \item $M$ treats each $p\in P$ as its individual constant and it interprets $p$ as the element $p$ in $P$; there are no other individual constants in the vocabulary of $M$.
 \item The relation symbols of $M$ are exactly those that occur in $S$.
 \end{itemize}
\item An \emph{override function} for $S$ is an assignment of truth values \Top, \Bot\ to \s\ formulas of the form
\begin{align*}
& A\lor B,\ \ A\to B
&&\text{where }A\ne B\quad\text{or}\\
& \a x A(x),\ \ \e x A(x)
&&\text{where $x$ has free occurrences in } A(x). \qefhere
\end{align*}
\end{itemize}
\end{definition}
\noindent
Beware that the free variables in $S$, if any, are treated as constants in $M$ because they are members of $P$.

\begin{definition}\label{d:val}
Let $S$ be a set of formulas, $M$ a standard structure for $S$, and $O$ an override function for $S$.
By induction on formula $X\in\s$, we define when $M$ \emph{\omodels} $X$, symbolically $M\models_O$ $X$ or simply \M{X} when $O$ is clear from the context.
\begin{enumerate}[leftmargin=30pt]

\item[(a)] If $X$ is atomic, \M{X} if $M$ models $X$ in classical logic.
    (In particular, \M{\top} but $M\nvDash \bot$.)
\item[(\Land)] \M{A\land B} if \M{A} and \M{B}.
\item[(\Lor)] \M{A\lor A} if $M\vDash A$;\\ if $A\ne B$, then\\
 \M{A\lor B} if \M{A} or \M{B} or $O(A \lor B) = \top$.
\item[(\To)] \M{A\to A};\\ if $A\ne B$ then\\
 \M{A\to B} if \M{B} or \big($M\nvDash A$ and $O(A \to B) = \top$\big).
\item[(\a)] If $x$ isn't free in $A$ then \M{\a x A} if \M{A};\\ otherwise
    \M{\a x A(x)} if \M{A(t)} for all $t\in P$ and $O(\a x A(x)) = \top$.
\item[(\e)] If $x$ isn't free in $A$ then \M{\e x A} if \M{A};\\ otherwise
    \M{\e x A(x)} if \M{A(p)} for some $p\in P$ or $O(\e x A(x)) = \top$.
\end{enumerate}
$M$ \emph{\omodels} a set $\H\subseteq S$ of formulas, symbolically $M\models_O \H$, if it \omodels\ every formula in \H.
\H\ \emph{semantically yields} a formula $X\in\s$, symbolically $\H\models X$, if $M\models_O \H$ implies $M\models_O X$ for all $M$ and $O$. \qef
\end{definition}

To explain the definition of $O$\!-modeling, think of the standard order $\Bot<\Top$ in the Boolean algebra \set{\top,\bot}.
The $\models_O$ relation is like the relation $\models$ of classical logic except that, in four cases, the values of $O$ can override the classical truth values.
\Top\ can be changed to \Bot\ in the cases (\To,\a) where the classical truth values are defined as ``the greatest such that \dots'', and
\Bot\ can be changed to \Top\ in the cases (\Lor,\e) where the classical truth values are defined as ``the least such that \dots''.

\begin{theorem}[Soundness]
\label{t:sound}
\Quad $\H\entails Q \implies \H\models Q$ \\
for any set \H\ of formulas and any formula $Q$.
\end{theorem}

\begin{proof}
Suppose $\H\entails Q$.
Let $S = \H\cup\set{Q}$ and $P = \Pars{S}$.
Then there is a minimal derivation \D\ of $Q$ from \H.
By Proposition~\ref{p:plocal}, we may assume that all parameters in \D\ are local.
By Theorem~\ref{t:local}, all \D-formulas belong to the closure \s\ of $S$.

Let $M, O$ range over standard structures for $S$ and override functions for \s, respectively.
Suppose that $M$ \omodels\ \H.
By induction on \D, from the leaves to the root, we prove that $M$ \omodels\ every node label $L(a)$ and in particular the label $Q$ of the root.

First suppose that $a$ is a leaf.
If $L(a)$ is a hypothesis, the claim holds by the choice of $M, O$. If $L(a)$ is an axiom, use clauses~(a) and (\To) in the above definition.

Next suppose that $a$ isn't a leaf.
A number of cases arise depending on what rule $R(a)$ is.
All cases are simple.
We consider one of the two \Lor I cases and the cases \Bot E, \To E,
\a E, and \e I.

\smallskip\noindent
Case $A\bs A\lor B$.
By the induction hypothesis, \M{A}.
By clause~(\Lor), \M{A\lor B} whether $A,B$ are equal or not.

\smallskip\noindent
Case $\bot\bs A$. This case does not arise.
Indeed, by induction hypothesis, \M{\bot}, but by the clause (a) of the definition of $\models$ relation, $M\nvDash\bot$.

\smallskip\noindent
Case $\seq{A, A\to B}\bs B$. By the induction hypothesis, \M{A} and \M{A\to B}. By clause~(\To), \M{B} whether $A,B$ are equal or not.

\smallskip\noindent
Case $\a x A(x)\bs A(t)$. By the induction hypothesis, \M{\a x A(x)}. Use clause~(\a).

\smallskip\noindent
Case $A(t)\bs\e x A(x)$. By the induction hypothesis, \M{A(t)}. Use clause~(\e).
\end{proof}

We mentioned in \S\ref{sb:disc}.A that $\seq{(A\to B), (B\to C)}\bs (A\to C)$ fails in QPL.
Now we can easily prove that.

\begin{corollary}\label{c:trans}
Let $A,B,C$ be distinct nullary relation symbols (which may be viewed as propositional variables).
Then formulas $A\to B$ and $B\to C$ do not entail $A\to C$.
\end{corollary}

\begin{proof}
Let $S$ comprise the three implications.
Let $M$ be a standard structure for $S$ where $A, B, C$ are false,
and let $O$ be an override function for $S$ that assigns \Top\ to $A\to B$ and $B\to C$ and assigns \Bot\ to $A\to C$.
Then $M$ $O$-models $A\to B$ and $B\to C$ but not $A\to C$.
By the soundness theorem, $A\to B$ and $B\to C$ do not entail $A\to C$.
\end{proof}

\begin{theorem}[Completeness]
\label{t:complete}
Suppose that $\H\not\vdash Q$.
There exist a structure $M$ and override function $O$ such that $M$ \omodels\ \H\ but not $Q$.
\end{theorem}

\begin{proof}
Let $Y$ range over formulas in \s.
The desired $M$ and $O$ are determined by the following conditions:
\begin{itemize}
\item $M$ is standard, and an atomic sentence holds in $M$ if and only if $\H\entails Y$, and
\item if $O(Y)$ is defined then $O(Y) = \top$ if and only if  $\H\entails Y$.
\end{itemize}
It suffices to prove
\begin{equation}\label{sem}
 M \models_O Y \iff \H \vdash Y
\end{equation}
for all formulas $Y$. We prove the claim by induction on the number $n$ of binary connectives and quantifiers in $Y$.
In the process, we use the definition of $O$, appropriate clauses in Definition~\ref{d:val}, and the induction hypothesis (in the induction step).

\smallskip\noindent
\texttt{Induction base:} $Y$ is \Top\ or atomic.
Use the definition of $M$.

\smallskip\noindent
\texttt{Induction step} splits into several cases depending on the form of $Y$.
We consider a few of them; the other cases are even easier.

\smallskip\noindent
Case $A\lor B$
where $A\ne B$. If \M{A} or \M{B}, then both sides of \eqref{sem} are true; otherwise
\[\M{(A\lor B)} \iff  O(A\lor B) = \Top \iff \H\vdash A\lor B. \]

\smallskip\noindent
Case $A\to B$. If $\M{B}$ then by induction hypothesis $\H\vdash B$,
and both sides of \eqref{sem} are true.
If $M\nvDash{B}$ but $\M{A}$, then $\H\not\vdash B$ but $\H\vdash A$, and both sides of \eqref{sem} are false.
If $M\nvDash{B}$ and $M\nvDash{A}$ then
\[
\M{(A\to B)} \iff O(A\to B)=\Top \iff \H\vdash{A\to B}.
\]

\smallskip\noindent
Case $\a x A(x)$
where $x$ has free occurrences in $A(x)$.
If \M{A(t)} for all $t\in\Pars{\s}$ then
\[
 \M{\a x A(x)} \iff O(\a x A(x)) = \Top
             \iff \H\vdash \a x A(x).
\]
Otherwise $M\nvDash A(t)$ for some $t\in\Pars{\s}$.
Then $\H\nvdash A(t)$ for that $t$, and both sides of \eqref{sem} are false.

\smallskip\noindent
Case $\e x A(x)$ where $x$ has free occurrences in $A(x)$.
If \M{A(t)} for some $t\in \Pars{\s}$, then $\H\entails A(t)$ for that $t$, and both sides of \eqref{sem} are true.
Otherwise $M\nvDash A(t)$ for all $t\in \Pars{\s}$.
Then $\H\nvdash A(t)$ for all $t\in \Pars{\s}$, and
\[
 \M{\e x A(x)} \iff O(\e x A(x)) = \Top
             \iff \H\vdash \e x A(x). \qedhere
\]
\end{proof}

\section{QPL entailment: upper bounds}
\label{s:upper}

The entailment problem for QPL is to decide, given a finite set \H\ of hypotheses and a query $Q$, whether \H\ entails $Q$.
Let $S = \H\cup\set{Q}$.
By the locality theorem, Theorem~\ref{t:local}, if $\H\entails Q$ then there is a derivation \D\ of $Q$ from \H\ where all \D-formulas are in the closure \s\ of $S$.

\begin{lemma}
The instance $\H,Q$ of the entailment problem for QPL is solvable in time $O\big(\lh{\s}^2\big)$.
\end{lemma}

\begin{proof}
The desired algorithm, let's call it Alg, systematically computes $D = \set{X\in\s: \H\entails X}$ until (i)~it encounters $Q$, in which case $\H\entails Q$, or (ii)~the computation of $D$ is completed without encountering $Q$, in which case $\H \nvdash Q$.

Alg starts with an auxiliary computation:
Walk through \s\ and, for each formula $X$ of the form $qx A(x)$ compute the set \Sub{X} of the substitution instances $A(p)$ where $p\in P$ (and $p$ is substitutable for $x$ in $X$).
After that, Alg initializes $D$ to $\H\cup\A$ where $\A$ is the set of QPL axioms in \s; let $N = |\s - (\H\cup\A)|$.

Then Alg traverses \s\ at most $N$ times.
It starts each of these rounds by checking whether $Q$ belongs to $D$. If yes, Alg stops and declares that \H\ entails $Q$.
If not, Alg attempts to expand $D$ using QPL rules. If no formula has been added to $D$ during the round, then Alg stops and declares that \H\ does not entail $Q$.
Otherwise, Alg proceeds to the next round.

It remains to describe how Alg attempts to expand $D$.
For each $X\in\s$ in turn, Alg executes as follows%
\footnote{In the last two cases below, we suppose the data structure allows computation of unions and intersections and checking membership and emptiness of sets in constant time.
Without such a data structure, the time bound would be $O(\lh{\s}^3)$ rather than $O(\lh{\s}^2)$, which would not affect Theorem~\ref{t:upper}.}.

\noindent
Case $X = A\land B$.
If $A,B\in D$, put $X$ into $D$ (which does not change $D$ if $X$ is already there).
If $X\in D$, put $A,B$ into $D$.

\noindent
Case $X = A\lor A$.
If one of \set{A, A\lor A} is in $D$, put the other into $D$.

\noindent
Case $X = A\lor B$ where $A\ne B$.
If one of \set{A, B} is in $D$, put $X$ into $D$.

\noindent
Case $X = A\to B$.
If $B\in D$, put $X$ into $D$.
If $A, X\in D$, put $B$ into $D$.

\noindent
Case $X = \bot$.
If $\Bot\in D$, put $Q$ into $D$ and declare that \H\ entails $Q$.


\noindent
Case $X = \a x A(x)$. 
If $X\in D$, execute $D := D \cup \Sub{X}$.

\noindent
Case $X = \e x A(x)$. 
If $D\cap\Sub{X}\ne\emptyset$, put $X$ into $D$.

Obviously, Alg works as intended and runs in time $O\big(\lh{\s}^2\big)$.
\end{proof}

\begin{theorem}\label{t:upper}
The QPL entailment problem is solvable in time $O(n^n)$.\\
In the case that the number of quantifications in $S$ is bounded, the QPL entailment problem is solvable in polynomial time.
\end{theorem}

\begin{proof}
Both claims follow from the lemma above and Proposition~\ref{p:closuresize}.
\end{proof}

\section{Lower bounds}
\label{s:lower}

Our goal in this section is to prove that the complexity bounds
of the previous section are optimal in a sense.
We shall show that the QPL entailment problem is EXPTime hard in general and is PTime hard  when the number of quantifiers in the input is bounded.
Here EXPTime$= \bigcup_{d>0}\text{Time}(2^{n^d})$, and PTime is of course polynomial time.

Our results will actually be stronger in two respects. First, we shall
establish that these hardness results hold even when the inputs
$(\H,Q)$ to the entailment problem are restricted to allow only
universal Horn sentences as the hypotheses in \H\ and to allow only
$\bot$ as the query $Q$.  In this connection, we take \emph{Horn
  clauses} to be iterated implications of the form
\begin{equation}\label{iform}
A_1\to\big(A_2\to\dots\to(A_m\to B)\cdots\big)
\end{equation}
where the $A_i$'s and $B$ are atomic formulas.
A \emph{universal Horn formula} is obtained by prefixing a Horn clause with some number (possibly zero) of universal quantifiers and thus binding some number of the free variables of the clause.
See below a discussion of the more familiar alternative forms
\begin{equation} \label{cform}
(A_1\land A_2\land\dots\land A_m)\to B
\end{equation}
and
\begin{equation} \label{dform}
   (\neg A_1)\lor(\neg A_2)\lor\dots\lor(\neg A_m)\lor B
\end{equation}
of Horn clauses.

Second, we shall establish our results not only for QPL but for a wide spectrum of logics (with or without function symbols of positive arity).


Let \L\ be a sublogic of classical first-order logic (as defined in \S\ref{sb:entail}).

\begin{lemma}[Conservative extension]\label{l:conserve}\mbox{}
\begin{enumerate}
\item Suppose that \L\ admits primal rules \To E and \a E.
Then a set \H\ of universal Horn formulas entails \Bot\ in \L\ if and only if \H\ entails \Bot\ in classical logic.
\item Suppose \L\ admits primal rule \To E.
Then a set \H\ of Horn clauses entails \Bot\ in \L\ if and only if \H\ entails \Bot\ in classical logic.
\end{enumerate}
\end{lemma}

\begin{proof}\mbox{}

\noindent
(1) The ``only if'' direction holds because \L\ is a sublogic of classical logic.
To prove the ``if'' direction, suppose that $\bot$ is not \L-derivable from \H.
We shall show that \H\ is satisfiable (in the sense of classical logic)  and therefore cannot classically entail $\bot$.

As before, let $P$ be  the set \Pars{\H} of parameters
occurring in $\H$ unless this set is empty, in which case $P$ comprises a single constant.
According to Definition~\ref{d:standard}, a standard structure for \H\ is a structure where the vocabulary consists of the elements of $P$ and the relation symbols in $\H$, the universe is $P$, and each $p\in P$ is treated as a constant and interpreted as itself.
Let $M$ be the standard structure where an atomic sentence $R(p_1,\dots,p_j)$ is true if and only if it is \L-derivable from $\H$. Recall that the free variables in \H\ are constants in $M$.

It remains to show that $M$ satisfies $\H$. To this end, consider any formula in $\H$ and eliminate all the universal quantifiers using parameters in $P$.
We shall check that the resulting $P$-instance
\[
A_1(\vec p)\to(A_2(\vec p)\to\cdots(A_m(\vec p)\to B(\vec p))\cdots)
\]
of the formula is true in $M$. Note that any such instance is \L-derivable from $\H$ because \L\ has the rule of universal instantiation.

If any of the antecedents $A_i(\vec p)$ in our instance is false in
$M$, then the whole instance is true. So assume that all
these antecedents are true. Since they are atomic, they are \L-derivable from $\H$ (by definition of our structure), and
therefore so is $B(\vec p)$ by repeated application of modus ponens,
which is available in \L. In particular, $B$ is not $\bot$, since
we've assumed that $\bot$ is not \L-derivable from $\H$. So
$B(\vec p)$ is atomic and \L-derivable from \H\ and therefore true in our structure. So our instance is true, as required.

\smallskip\noindent
(2) The proof of Claim~(2) is a special case of the proof of  Claim~(1) where there are no universal quantifiers in \H.
\end{proof}


Let HE(\L), an allusion to ``Horn Entailment (of \Bot) in \L'', be the problem to decide whether a given set \H\ of universal Horn formulas entails \Bot\ in \L.
For any non-negative integer $k$, let HE($\L,k$) be the restriction
of HE(\L) to instances \H\ that contain at most $k$ occurrences
of quantifiers.

\begin{theorem} \label{t:hard2}\mbox{}
  \begin{enumerate}
  \item If \L\ admits primal rules \To E and \a E, then HE(\L) is EXPTime hard.
  \item If \L\ admits primal rules \To E, then HE(\L,0) is PTime hard.
  \item If \L\ admits primal rules \To E and \a E, then HE($\L,k$) is PTime hard for all $k\ge0$.
  \end{enumerate}
\end{theorem}

\begin{proof}
By the conservative extension lemma, it suffices to prove the theorem for classical logic.

\smallskip\noindent
(1) The entailment problem for classical logic is the complement
of the satisfiability problem, and the result we need is essentially proved in \cite[Theorem~4.5]{DEGV}.
Strictly speaking, the cited proof concerns a version of HE(classical logic) in which the query $Q$ can be an arbitrary atomic sentence rather than \Bot.
That makes no difference because, in classical logic, \H\ entails $Q$ if and only if $\H\cup\set{Q\to \bot}$ entails \Bot.

\smallskip\noindent
(2) The proof of PTime hardness for this case is similar to the proof of Claim~(1), except that this time around we cite Theorem~4.2 in \cite{DEGV} rather than Theorem~4.5.

\smallskip\noindent
(3) It suffices to consider the case $k=0$, since this case is trivially reducible to any case with larger $k$.
Hence (3) follows from (2).
\end{proof}

\begin{corollary}\label{c:lowerbounds}\mbox{}
\begin{enumerate}
\item The QPL entailment problem is EXPTime hard.
\item The entailment problem for the original primal logic, for the propositional fragment of QPL, and for all other logics in \S\ref{s:lin} is PTime hard.
\item For any non-negative integer $k$, the restriction of the QPL entailment problem to the case of $\le k$ quantifiers is PTime hard.
\end{enumerate}
\end{corollary}

\smallskip\noindent\texttt{Discussion.}
We consider what would happen to our results if we used the \emph{conjunction form} \eqref{cform} or the  \emph{disjunction form} \eqref{dform} of Horn clauses rather than the
the \emph{implication form} \eqref{iform}.

Had we used the conjunction form, the proof of the conservative extension lemma would need that \L\ also admits the conjunction introduction rule \Land I.
Specifically, when proving that \H\ is satisfied by the
structure $M$ we constructed, we would need that, if all of
$A_1,\dots,A_m$ are derivable from \H, then these together with
$(A_1\land A_2\land\dots\land A_m)\to B$ entail $B$. The proof of that
would involve applying $\land$I repeatedly to derive
$A_1\land A_2\land\dots\land A_m$ and then applying $\to$E to obtain
$B$.
Our reason for preferring the implication form over the conjunction
form in this section was merely to avoid the assumption that \L\ admits \Land I.

The situation for the disjunction form is considerably worse.
The least of our problems is that the disjunction form uses negation,
which is not among the connectives of primal logic.
This is solved by our convention, common in intuitionistic logic,
that $\neg A$ is an abbreviation for $A\to\bot$.

A more significant problem is that the disjunction form of a Horn
clause, though classically equivalent to the implication and
conjunction forms, is intuitionistically strictly stronger. Consider,
for example, the Horn clause $A\to A$, which is in both implication
and conjunction form. It is intuitionistically valid, but its
disjunction form, $(\neg A)\lor A$, expresses the law of the excluded
middle, which is of course not intuitionistically valid.

Despite this disagreement between classical and intuitionistic logic
about the status of disjunction forms relative to the other two forms,
the decision problems HE(classical logic) and HE(intuitionistic
logic) coincide, even if ``Horn clause'' is understood as meaning the
disjunction form. That is, given a set \H\ of universal Horn
formulas in disjunction form, $\bot$ is intuitionistically derivable
from \H\ if and only if it is classically derivable from \H.
The ``only if'' direction is obvious as classical logic subsumes
intuitionistic logic.

To prove the ``if'' direction, suppose \H\
entails $\bot$ classically.  Then so does the set $\H'$ obtained
by converting every sentence \eqref{dform} in \H\ to the classically equivalent implication form \eqref{iform}. By Lemma~\ref{l:conserve}, $\H'$ entails $\bot$ intuitionistically also. But then so
does $\H$, because \eqref{dform} intuitionistically entails \eqref{iform}.
To prove that entailment, it suffices to show that $(\neg A \lor C)$, i.e. $((A\to\bot) \lor C)$, entails $(A\to C)$.
By the intuitionistic \Lor E rule, it suffices to show that each disjunct entails the implication, which is rather easy thanks to the intuitionistic \To I rule.

Unfortunately, this equivalence between HE(classical logic) and
HE(intuitionistic logic) gives no information about primal logic;
these two decision problems are not equivalent to HE(QPL) when Horn
clauses are in disjunction form. For example, the pair of implication
forms $\{A, A\to\bot\}$ entails $\bot$ by a single application of
$\to$E, but the corresponding pair of disjunction forms, $\{A,(\neg
A)\lor\bot\}$, i.e., $\{A,(A\to\bot)\lor\bot\}$, does not entail $\bot$
in QPL. Another example is that $\bot$ is derivable in primal logic
from $\{\top\to\bot\}$ but not from $\{(\neg\top)\lor\bot\}$, i.e.,
$\{(\top\to\bot)\lor\bot\}$.  In both cases, the non-derivability
claims are easy to verify using the semantics in Section~\ref{s:sem}; the
key is that an override function can give the value $\top$ to any disjunction of distinct formulas. \qef

\appendix
\section{Natural deduction calculi\\
for intuitionistic and classical logics}
\label{app:a}

We start with a calculus for intuitionistic logic:

\begin{align*}
&\hspace{35pt}\text{Rules}
&\hspace{35pt}&\text{Metarules}
\\[10pt]
\top\text{I}
&\hspace{35pt}\top
&& \emptyset \implies \top
\\[5pt]
\land\text{I}
&\hspace{35pt}\fr{A\quad B}{A\land B}
&& A,B \implies A\land B
\\[9pt]
\land\text{E}
&\hspace{35pt}\fr{A\land B}A\comma\ \fr{A\land B}B
&& \hspace{-30pt}
   A\land B\implies A,\ A\land B\implies B
\\[9pt]
\lor\text{I}
&\hspace{35pt}\fr{A}{A\lor B}\comma\ \fr{B}{A\lor B}
&& \hspace{-30pt}
   A\implies A\lor B,\ B\implies A\lor B
\\[9pt]
\lor\text{E}
&\hspace{35pt}
 \fr{A\lor B\quad\raisebox{7pt}{$(A)$}C\quad\raisebox{7pt}{$(B)$}C} C
&& \fr{A,\G\implies C\quad B,\G\implies C}{A\lor B,\G\implies C} \\[9pt]
\to\text{I}
&\hspace{35pt}
  \fr{\raisebox{7pt}{$(A)$}B}{A\to B}
&& \fr{A,\G\implies B}{\G\implies (A\to B)}
\\[9pt]
\to\text{E}
&\hspace{35pt}\fr{A\quad A\to B}B
&& A, A\to B \implies B
\\[9pt]
\bot\text{E}
&\hspace{35pt}\fr\bot A
&& \bot \implies A
\\[9pt]
\a\text{I}
&\hspace{35pt}\fr{A(p)}{\a x A(x)}
&&\fr{\G\implies A(p)}{\G\implies\a x A(x)}\\
&&&\hspace{-30pt}
\text{\footnotesize where parameter $p$ is not exposed in $\G,\a x A(x)$} \\[5pt]
\a\text{E}
&\hspace{35pt} \fr{\a x A(x)}{A(t)}
&& \a x A(x) \implies A(t)
\\[9pt]
\e\text{I}
&\hspace{35pt} \fr{A(t)}{\e x A(x)}
&& A(t) \implies \e x A(x)
\\
\e\text{E}
&\hspace{35pt} \fr{\e x A(x)\quad \raisebox{10pt}{$(A(p))$}B} {B}
&& \fr{A(p),\G\implies B}{\e x A(x),\G\implies B}
\\
&&&\hspace{-30pt}
\text{\footnotesize where parameter $p$ is not exposed in $\G,\e x A(x),B$}
\end{align*}

\medskip
A classical calculus is obtained from the intuitionistic one by adding an axiom schema $A \lor \neg A$ where $\neg A$ abbreviates $A\to\bot$.

\section{Hilbert-type calculus for quantified primal logic}
\label{app:b}

\begin{align*}
\top\text{I} \hspace{20pt}
&\hspace{10pt}\top
\\[15pt]
\land\text{I}\hspace{20pt}
&\fr{A\quad B}{A\land B}\hspace{80pt}
&\land\text{E}\hspace{20pt}
&\fr{A\land B}A\comma\ \fr{A\land B}B
\\[20pt]
\lor\text{I}\hspace{20pt}
&\fr{A}{A\lor B}\comma\ \fr{B}{A\lor B}\hspace{80pt}
&\lor\text{E}\hspace{20pt}
&\hspace{10pt}\fr{A\lor A} A
\\[20pt]
\to\text{I}\hspace{20pt}
&\fr B{A\to B}\comma\ A\to A\hspace{80pt}
&\to\text{E}\hspace{20pt}
&\fr{A\quad A\to B}B
\\[20pt]
&&\bot\text{E}\hspace{20pt} &\hspace{20pt}\fr\bot A
\\[25pt]
\a\text{I}\hspace{20pt}
&\hspace{10pt}\fr{A}{\a x A}\hspace{70pt}
&\a\text{E}\hspace{20pt}
&\hspace{10pt}\fr{\a x A(x)}{A(t)} \\
&\hspace{-15pt} \text{where $x$ isn't free in $A$}
\\[25pt]
\e\text{I}\hspace{20pt}
&\hspace{10pt}\fr{A(t)}{\e x A(x)}\hspace{70pt}
&\e\text{E}\hspace{20pt}
&\hspace{10pt}\fr{\e x A} {A} \\
&&&\hspace{-15pt}\text{where $x$ isn't free in $A$}
\end{align*}

\section{Primal logic with full-fledged quantifier rules}

Of the four standard intuitionistic/classical quantifier rules in Appendix~A, our quantified primal logic QPL uses, as is, only two rules, \a E and \e I. The other two rules, \a I and \e E, are crippled by Flatting.
In this appendix, we show that enriching the original primal logic with all four standard quantifier rules (or just three, namely \a E, \e E, \e I) leads to an undecidable entailment problem.

\ph{Logics.}
Let \L\ be a sublogic of classical logic such that \L\ admits the intuitionistic natural-deduction rule schemas
\[ \Land E,\ \To E,\ \a E,\ \e E,\ \e I \]
of Appendix~A.
Four of these five rule schemas, all but \e E, are Hilbert-style.
According to \S\ref{sb:hilbert}, \L\ admits a Hilbert style rule $\seq{A_1, \dots, A_k} \bs B$ if the premises $A_1, \dots, A_k$ entail the conclusion $B$ in \L.
The situation is a little bit more involved with \e E.
Consider an instance
\[ \Big\langle\e x A(x),\ \raisebox{7pt}{$(A(p)$)}B\Big\rangle\bs B \]
of \e E where $p$ is a constant; only such instances of \e E are relevant for our purposes.
\L\ \emph{admits} the displayed instance if the following holds in \L\ for  any set $\G,B$ of formulas: Formulas $\e x A(x),\G$ entail $B$ if formulas $A(p),\G$ entail $B$ and $p$ does not occur in $\G,B,\e x A(x)$.

We assume without loss of generality that (the decimal notation for) every nonnegative integer is an individual constant in \L.

\noindent\texttt{Two-register machines.}
A two-register machine $M$, as in \cite[p.~27--28]{G00}, is a deterministic computation device similar to a Turing machine, but instead of a tape it has two registers, register~1 and register~2, containing nonnegative integers.
A \emph{configuration} of $M$ is a triple $(i,m,n)$ where $i$ is the (control) state and $m,n$ are the contents of registers~1 and 2 respectively.
The initial configuration is $(0,0,0)$, and the only halting state is 1.
Each state $i$ of $M$ is associated with an increment or decrement instruction.

\begin{itemize}
\item An \emph{increment} instruction $(r,j)$:
Increment register~$r$ (by one) and go to state $j$.
\item A \emph{(conditional) decrement} instruction $(r,j,l)$:
First test whether register~$r$ contains zero.
If yes, then go to state $j$.
If not, then decrement register~$r$ (by one) and go to state $l$.
\end{itemize}
The halting problem for two-register machines is undecidable \cite{Minsky,S&S}.

\ph{Logic description.}
We adapt to our purposes the description of two-register machines in the proof of Theorem~2.1.15 in \cite{G00} (from \cite{Aanderaa} and \cite{Boerger}.)

A two-register machine $M$ is described by means of a sentence $\phi_M$ of the form
\[ (\a x\e x' Sxx')  \land K_000 \land
\Big(\a xx'y\big(Sxx'\to\bigwedge_i\d_i(x,x',y)\big)\Big) \]
using binary relations $S$ and $K_i$ for every state $i$ of $M$.

The intention behind the first conjunct is that $x'$ is the successor of $x$.
It would be natural to use a unary successor function $'$, but we work without function symbols of positive arity.\

The intended meaning of $K_imn$ is that the configuration $(i,m,n)$ of $M$ is reachable.
The second conjunct reflects the fact that the initial configuration $(0,0,0)$ is reachable.

The third conjunct of $\phi_M$ reflects the one-step action of $M$.
If the state $i$ instruction is an increment instruction $(r,j)$, then
\[\d_i(x,x',y) = \begin{cases}
K_ixy \to K_jx'y &\text{if }r=1,\\
K_iyx \to K_jyx' &\text{if }r=2,
\end{cases}\]
and if the state $i$ instruction is a decrement instruction $(r,j,l)$, then
\[\d_i(x,x',y) = \begin{cases}
(K_i0y \to K_j0y) \land (K_ix'y \to K_lxy)  &\text{if }r=1,\\
(K_iy0 \to K_jy0) \land (K_iyx' \to K_lyx) &\text{if }r=2.
\end{cases}\]

Let Halting abbreviate $\e xy K_1(x,y)$.

\begin{proposition}
The following statements are equivalent.
\begin{enumerate}
\item $\phi_M$ entails Halting in \L.
\item $\phi_M$ entails Halting in classical logic.
\item $M$ halts.
\end{enumerate}
\end{proposition}

\begin{remark}
One can describe a two-register machine using no constants.
Whether we speak about classical logic or \L, $\e u\, \phi_M(0\mapsto u)$ entails Halting in \L\ if and only if $\phi_M$ entails Halting, because both logics admit \e E and \e I.
Note, in connection with \e E, that 0 does not occur and thus is not exposed in Halting. \qef
\end{remark}

\begin{proof}[Proof of the proposition]
(1)$\implies$(2) because \L\ is a sublogic of classical logic.

(2)$\implies$(3). Suppose (2).
Consider the structure $X$ of the vocabulary of $\phi_M$ on nonnegative integers where the constant 0 denotes the number zero, $Sxy$ is the relation $y = x+1$, and $K_i(m,n)$ is true if and only if $M$'s configuration $(i,m,n)$ is reachable from the initial configuration $(0,0,0)$.
Clearly $X$ is a model for $\phi_M$.
By (2), Halting holds in $X$, so $(1,m,n)$ is reachable for some $m,n$.
Thus, $M$ halts.

(3)$\implies$(1). Assume that $M$ halts after $T$ steps and consider the run $C_0, C_1, \dots, C_T$ of $M$, where $C_t$ is a configuration $(i_t,m_t,n_t)$, from the initial configuration $C_0 = (0,0,0)$ to the final configuration $C_T = (1,m_T,n_T)$.
In the rest of the proof we work with \L\ and prove (1).

Recall that (decimal notations of) nonnegative integers are individual constants of \L.
$\phi_M$ entails $\a x\e x' Sxx'$ which entails $\e x' S0x'$; $\phi_M$ entails Halting if $\phi_M$ and $\e x' S0x'$ entail Halting, hence if $\phi_M$ and $S$01 entail Halting.
Similarly, $\phi_M$ and $S$01 entail Halting if $\phi_M$, $S$01, and $S$12 entail Halting, and so on.
It suffices to prove that $\phi_M$ and $S$01, $S$12, \dots, $S(T-1)T$ entail Halting.
Let \H\ be the hypotheses $\phi_M$ and $S$01, $S$12, \dots, $S(T-1)T$.

Further, it suffices to prove that, for every nonnegative integer $t\le T$,

\noindent
\textbf{Claim}~$t$:\qquad $\H\vdash K_{i_t}(m_t,n_t)$,

\noindent
so that in particular, for $t=T$, we have that \H\ entails $K_1(m_T,n_T)$ which entails Halting by \e I.

We prove Claims~$t$ by induction. Claim~0 is obvious.
Suppose that $t<T$ and that Claim~$t$ has been proven.
It remains to prove Claim~$(t+1)$.
To simplify notation, we abbreviate $i_t,m_t,n_t$ to $i,m,n$ respectively. Let $I$ be the state $i$ instruction.
By symmetry, we may assume that $I$ works with Register~1.
Three cases arise: $I$ is an increment instruction, $I$ is a decrement instruction and $m=0$, and $I$ is a decrement instruction and $m>0$.
We consider only the third case where $I$ has the form $(1,j,l)$ and $m = k+1$ where $k\ge0$.

By \Land E, $\phi_M$ entails its conjuncts, and, by \a E, the third conjunct  entails the instance of itself where $x,x',y$ are instantiated to $k, m, n$: $Skm \to \bigwedge_h \d_h$.
Recall that $Skm$ is one of the hypotheses in \H.
Hence \H\ entails $\bigwedge_h \d_h$ and, by \Land E, entails $\d_i(k,m,n)$, i.e.
\[(K_i0n \to K_j0n) \land (K_imn \to K_lkn) \]
which allows us to derive $K_imn \to K_lkn$.
By the induction hypothesis, \H\ entails $K_imn$.
By \To E, \H\ entails $K_lkn$ which is Claim~$(t+1)$.
\end{proof}

Since the halting problem for two-register machines is undecidable, the proposition implies the following theorem.

\begin{theorem}\label{t:undec}
The entailment problem for \L\ is undecidable.
\end{theorem}

\subsection*{Acknowledgment}
We thank Lev Beklemishev for useful discussions.

\end{document}